\documentclass[aip,cha,
amsmath,amssymb,amsfonts,
reprint,onecolumn,tightenlines,11pt,nofootinbib,
groupedaddress,
a4paper
]{revtex4-1}
\usepackage{amsthm}
\usepackage{hyperref}
\usepackage[utf8]{inputenc}
\usepackage{bbm,braket}
\theoremstyle{definition}

\newtheorem{definition}{Definition}
\newtheorem{lemma}[definition]{Lemma}
\newtheorem{proposition}[definition]{Proposition}

\providecommand{\abs}[1]{{\lvert{#1}\rvert}}
\providecommand{\norm}[1]{{\lVert{#1}\rVert}}
\providecommand{\openone}{{\mathbbm 1}}
\providecommand{\probability}{{\mathbbm P}}
\providecommand{\reals}{{\mathbbm R}}
\providecommand{\Rho}{{\mathrm P}}
\providecommand{\tsum}{{\textstyle \sum}}

\newcommand{\exbound}{{\partial^+}}
\newcommand{\exv}[1]{{\langle{#1}\rangle}}
\newcommand{\proj}[1]{{\ket{#1}\!\!\bra{#1}}}
\newcommand{\then}{{\rhd}}
\newcommand{\vtr}[1]{{\mathbf{#1}}}

\DeclareMathOperator{\cone}{cone}
\DeclareMathOperator{\diag}{diag}
\DeclareMathOperator{\id}{id}
\DeclareMathOperator{\tr}{tr}

\begin{document}
\title{
Sequences of projective measurements in generalized probabilistic models
}

\author{Matthias Kleinmann}
\affiliation{%
Naturwissenschaftlich-Technische Fakultät,
Universität Siegen,
Walter-Flex-Straße 3,
57068 Siegen,
Germany}%
\affiliation{%
Departamento de Matemática,
Universidade Federal de Minas Gerais,
Caixa Postal 702,
Belo Horizonte,
Minas Gerais 31270-901,
Brazil}

\begin{abstract}
We define a simple rule that allows to describe sequences of projective 
 measurements for a broad class of generalized probabilistic models.
This class embraces quantum mechanics and classical probability theory, but, 
 for example, also the hypothetical Popescu-Rohrlich box.
For quantum mechanics, the definition yields the established Lüders's rule, 
 which is the standard rule how to update the quantum state after a 
 measurement.
In the general case it can be seen as the least disturbing or most coherent way 
 to perform sequential measurements.
As example we show that Spekkens's toy model \cite{Spekkens:2007PRA} is an 
 instance of our definition.
We also demonstrate the possibility of strong post-quantum correlations as well 
 as the existence of triple-slit correlations for certain non-quantum toy 
 models.
\end{abstract}

\maketitle

\section{Introduction}
It is a fundamental property of quantum mechanics that any nontrivial 
 measurement disturbs the system it acts on.
This disturbance is responsible for very particular phenomena like the quantum 
 Zeno effect \cite{Peres:1995, Itano:1990PRA}, where the time-evolution of a 
 system is frozen due to repeated measurements, or the contextual behavior of a 
 quantum system \cite{Kirchmair:2009NAT}, where measurement outcomes depend on 
 the choice of previous compatible measurements.
Compared to the classical world, where a measurement---at least in 
 principle---may leave the system unchanged, this quantum property seems to be 
 very particular and at the same time very fundamental.

The most common formulation of the this disturbance is due to Lüders 
 \cite{Luders:1951APL, Luders:2006AP} and determines how the state of a system 
 changes after a measurement: $\rho\mapsto \Pi\rho \Pi / \tr(\rho \Pi)$.
But this is only one out of many possible state changes that may occur in an 
 experiment.
In the most general case the post-measurement state can be seen as the result 
 of a coherent evolution involving an auxiliary system and a destructive 
 measurement on that auxiliary system.
This fundamental result by Ozawa \cite{Ozawa:1984JMP, Heinosaari:2012} does, 
 however, not explain the special role of Lüders's rule.
Conversely, Ozawa's result gives a very particular model of a measurement and 
 one might argue that giving up Lüders's rule as a fundamental entity might 
 actually make too strong assumptions on the peculiarities of the measurement 
 process in quantum mechanics.

In this work we provide a very small set of assumptions that uniquely singles 
 out Lüders's rule within quantum mechanics on the one hand, and on the other 
 hand has many desirable properties when applied to hypothetical non-quantum 
 models.
These two aspects have been discussed for a long time \cite{Mielnik:1969CMP, 
 Araki:1980CMP, Niestegge:2008FPH, Ududec:2011FPH}, and some consensus seems to 
 exist that the mathematical concept of a filter is an appropriate approach.
We advertise that the axioms that we suggest here are significantly simpler 
 then those that have appeared before while at the same time they imply more 
 favorable physical properties.

We proceed as follows.
The introduction is completed by a detailed reminder on how post-measurement 
 states are treated in quantum mechanics, cf.\ Sec.~\ref{s1003}, and a summary 
 of the mathematical framework of ordered vector spaces in Sec.~\ref{s18520}, 
 enriched with examples in Sec.~\ref{s9680}.
In Sec.~\ref{s8520} we introduce the notion of projective, neutral, and 
 coherent $f$-compatible maps, the latter of which we propose as a generalized 
 definition of Lüders's rule.
We investigate fundamental properties of this definition and give examples, in 
 particular we study the case of quantum mechanics in Sec.~\ref{s15325}, a 
 large class of toy models in Sec.~\ref{s16520}, and the $n$-slit experiment in 
 Sec.~\ref{s26365}.
We conclude with a discussion of our findings in Sec.~\ref{s12113}.

\subsection{Quantum instruments}\label{s1003}
Before we start to formulate the behavior of measurement sequences in 
 generalized probabilistic models, let us first recall the established 
 formalism in quantum mechanics \cite{Heinosaari:2012}.

We consider a situation where first an observable $A$ and then an observable 
 $B$ is measured.
(In order to simplify the discussion, we assume that $A$ and $B$ have pure 
 point spectrum.)
The system subject to the measurements is initially described by a density 
 operator $\rho$ and the measurement of $A$ is assumed to have yielded the 
 result $a$.
With the spectral decomposition as $A= \sum_a a \Pi_a$, according to Lüders 
 \cite{Luders:1951APL, Luders:2006AP}, the expected value of $B$ is given by
\begin{equation}\label{e5921}
 \exv{B|A= a}_\rho= \tr[\Pi_a\rho \Pi_a B] / \tr(\rho \Pi_a)=
 \tr[\rho \phi_a(B)] / \tr[\rho \phi_a(\openone)].
\end{equation}
For the second equality we introduced the map $\phi_a\colon X\mapsto \Pi_a X 
 \Pi_a$, so that it becomes manifest that the conditioned expectation value on 
 the l.h.s.\ arises directly from the laws of conditional probabilities and the 
 quantum instrument
 $\mathcal I_\mathrm{L}\colon a\mapsto \phi_a$.
(In literature, the notion of a Lüders instrument has been established, but it 
 covers a broader set of instruments then those that follow Lüders's rule.)

The situation described in Eq.~\eqref{e5921} can be further formalized.
With the spectral decomposition $B=\sum b \Rho_b$, the probability to get 
firstly the outcome $a$ and then the outcome $b$ is
\begin{equation}
 \probability_\omega(\Pi_a\then \Rho_b)=
   \omega[\phi_a(\Rho_b)],
\end{equation}
 where $\omega\colon X\mapsto \tr(\rho X)$ is a way to write the quantum state 
 and $\Pi_a\then \Rho_b$ is the event ``$\Pi_a$ then $\Rho_b$.''

Depending on the experimental implementation, the actual instrument $\mathcal 
 I'$ will deviate from the instrument that has been described by Lüders.
But there is confidence that $\mathcal I_\mathrm{L}$ can be approximated to an 
 arbitrary precision, since on a formal level \cite{Ozawa:1984JMP} one can 
 implement $\mathcal I_\mathrm{L}$ by virtue of an ancilla system in a pure 
 state, an entangling unitary between the probe and ancilla system, and a 
 destructive measurement solely on the ancilla system.
This shows that $\mathcal I_\mathrm{L}$ can be implemented as an immediate 
 consequence of
\begin{itemize}\itemsep-.1em
\item[(i)]
independent pure state preparation $\rho\mapsto \rho\otimes \proj \psi$,
\item[(ii)]
unitary evolution,
\item[(iii)]
Born's rule, $\probability_\omega(A=a)= \omega(\Pi_a)$.
\end{itemize}

However, any instrument can be implemented with the ingredients (i)--(iii).
The question that drives our subsequent analysis is which of the properties of 
 the instrument $\mathcal I_\mathrm{L}$ corresponding to Lüders's rule are most 
 characteristic.
Within the framework of quantum mechanics there would be a variety of possible 
 characteristics that single out Lüders's rule and without comparing to other 
 possibilities, it would be difficult to argue in favor of one or another.
Our approach is to broaden the mathematical concepts, so that not only quantum 
 mechanics can be described but also a wider set of generalized probabilistic 
 models is covered.

\subsection{Positivity and generalized probabilistic models}\label{s18520}
Quantum events as well as classical events can be mathematical described by 
 ordered vector spaces.
This is based on the observation that the main characteristics of either theory 
 is dominated by the notion of positivity.
In particular in quantum mechanics, the (mixed) states are given by maps 
 $\omega\colon X\mapsto \tr(\rho X)$ which obey $\omega(\openone)= 1$ and 
 $\omega(F)\ge 0$ for all positive semi-definite operators $F$.
Conversely, a generalized measurement in quantum mechanics is a family of 
 positive semi-definite operators $(F_a)$ with $\sum_a F_a= \openone$.
The operators $F_a$ are then called effects.
This positivity structure is largely motivated from the probabilistic 
 interpretation $\probability_\omega(F_a)= \omega(F_a)$.
The class of models which follows a similar interpretation is captured by the 
 mathematical concept of an ordered vectors space.
In turn, the set of models that can be fitted into this mathematical concept 
 contains instances that are in conflict with the predictions of quantum 
 mechanics \cite{Popescu:1994FPH, Janotta:2012EPTS}.
For this reason, these models are called generalized probabilistic models.

We now discuss the mathematical concepts related to ordered vectors spaces 
 while in Sec.~\ref{s9680} we present explicit examples.
For a more verbose introduction into the mathematical concepts we particularly 
 recommend the introduction of Ref.~\onlinecite{Paulsen:2009IUMJ} and the books 
 by Alfsen [\onlinecite{Alfsen:1971}] and Paulsen [\onlinecite{Paulsen:2002}].
A real order unit vector space is a triple $(V,V^+,e)$, such that
\begin{itemize}\itemsep-.1em
\item[(i)]
$V$ is a real vector space (not necessarily finite-dimensional).
\item[(ii)]
$V^+\subset V$ is a cone, i.e., $V^+ + V^+= V^+= \reals^+ V^+$ and $V^+\cap 
 -V^+= \set 0$.
\item[(iii)]
$e\in V^+$ is an order unit, i.e., for any $x\in V$ there is an $r\in \reals^+$ 
 such that $re+ x\in V^+$.
\end{itemize}
We wrote $\reals^+$ for the set of non-negative reals.
It follows \cite{Paulsen:2009IUMJ} that $V^+-V^+= V$.
For two elements $x,y\in V$ the condition $x-y\in V^+$ defines a partial order 
 and one writes $x\ge y$.

The order unit $e$ is Archimedean provided that for any $x\in V$ the property 
 $x+\reals^+ e\subset V^+ \cup \set x$ implies $x\in V^+$.
This property in some sense requires that $V^+$ is ``closed.''
While we use this property merely for technical reasons, also note, that an 
 order unit vector space can always be modified in such a way that it has an 
 Archimedean order unit.
This Archimedeanization \cite{Paulsen:2009IUMJ} works by constructing the 
 ``closure'' of the cone and identifying operationally indistinguishable 
 elements.
These operations are physically benign and hence we only consider Archimedean 
 order unit vector (AOU) spaces.

We continue to fix notation.
Within the dual space $V^*= \set{\alpha\colon V\rightarrow \reals | \alpha 
 \text{ is linear}}$ the set
\begin{equation}
 \mathcal S(V,V^+,e)= \Set{\omega\in V^* | \omega(e)= 1 \text{ and } 
 \omega(V^+)\subset \reals^+ }
\end{equation}
 is the convex set of states and the definition
\begin{equation}
 \norm{x}= \inf \Set{r\in \reals^+| -re \le x \le re}
\end{equation}
 provides the order norm of $x\in V$.
It is convenient to define the set of effects, i.e., the convex set of positive 
 elements bounded by $e$,
\begin{equation}
 V_e^+= V^+\cap (e-V^+),
\end{equation}
 and to write for the normalized representatives of the extremal rays of $V^+$ 
 the symbol
\begin{equation}
 \exbound V^+
  =\Set{ f\in V^+ | \norm f= 1 \text{ and } ( 0\le g\le f \text{ implies }
  g\in \reals^+ f) }
\end{equation}
We occasionally construct $V^+$ from a finite set $\mathcal A\subset V$ of 
 extremal rays via
\begin{equation}
 \cone\mathcal A=
 \Set{ x\in V | x= \tsum_{a\in \mathcal A} r_a\, a \text{, where all }
 r_a\in \reals^+}.
\end{equation}

For two AOU spaces $(V,V^+,e)$ and $(W,W^+,e')$, a linear map $\phi\colon 
 V\rightarrow W$ is positive, provided that it maps positive elements to 
 positive elements, $\phi(V^+)\subset W^+$.
(When we let $\phi$ be a map, we always imply that $\phi$ is linear.)
If $\phi(e)= e'$ then $\phi$ is unital.
The spaces are order isomorphic, if there exists a positive unital bijection 
 $\psi\colon V\rightarrow W$ such that its inverse is also positive.

\begin{proposition}\label{p10788}
We recall three results from Ref.~\onlinecite{Paulsen:2009IUMJ}.
\begin{itemize}\itemsep-.1em
\item[(i)]
$f\in V^+$ if and only if $\omega(f)\ge 0$ for all $\omega\in \mathcal S$.
\item[(ii)]
If $f\in V^+$, then there exists a state $\omega\in \mathcal S$, such that 
 $\omega(f)= \norm f$.
\item[(iii)]
For $x\in V$ we have $-\norm x e\le x \le \norm x e$.
\end{itemize}
\end{proposition}

~

In principle one is free to choose the AOU space $(V,V^+,e)$ or the states 
 $\mathcal S\subset U$ with some embedding vector space $U$ as fundamental 
 object.
If $\mathcal S$ is fundamental, then \cite{Araki:1980CMP} we can define $V$ to 
 be the space of affine functions on $U$, let $V^+= \set{\xi \in V | 
 \xi(\mathcal S)\subset \reals^+}$, and choose $e$ with $e(\mathcal S)= \set 
 1$.
Since we do not want to make any particular point out of which space is 
 fundamental, we may assume that $V$ is reflexive, $V= V^{**}$.
By virtue of Proposition~\ref{p10788}~(i) this would imply that $(V,V^+,e)$ and 
 $[V^{**}, (V^{**})^+, e^{**}]$ are order isomorphic.

\subsection{Examples of ordered vectors spaces}\label{s9680}
The reason why AOU spaces are considered to be a good framework to describe 
 generalized probabilistic models is that classical events and quantum events 
 can be described by means of AOU spaces \cite{Ludwig:1987, Mittelstaedt:1998}.
For a recent introduction into the physical interpretation we refer to 
 Ref.~\onlinecite{Barnum:2011ENTS}.

\vspace{.5em}\noindent\textbf{Classical events.}
A set of discrete classical events---e.g.\ the outcomes when rolling a 
 dice---defines a so-called AOU lattice.
It is the $n$-fold Cartesian product of $(\reals,\reals^+,1)$, where $n$ is the 
 number of outcomes.
The set of states is given by the maps $\vtr v\mapsto \vtr p\cdot \vtr v$ with 
 $\vtr p_k\ge 0$ for all $k$, and $\sum_k \vtr p_k= 1$.
The order norm reads $\norm{\vtr v}= \max_k \abs{\vtr v_k}$, turning $V$ into 
 the Banach space $\ell_n^\infty$.

\vspace{.5em}\noindent\textbf{Quantum events.}
For quantum mechanics, we choose the bounded self-adjoint operators as vector 
 space $V$ and we identify $V^+$ to be the set of positive semi-definite 
 operators.
With the choice $e= \openone$ this forms an AOU space, cf.\ Theorem~{1.95} in 
 Ref.~\onlinecite{Alfsen:2001}.
The set of quantum effects is $V^+_e$.
The quantum states can be represented by the maps $X\mapsto \tr(\rho X)$ where 
 $\rho$ is positive semi-definite with $\tr\rho= 1$.
(For infinite-dimensional Hilbert spaces, however, not all functionals in 
 $\mathcal S$ can be written this way.)
The order norm $\norm X$ yields the operator norm of $X$ and the extremal set 
 $\exbound V^+$ is exactly the set of rank-one projections.

\vspace{.5em}\noindent\textbf{Dichotomic norm cones.}
A simple class of examples is constructed as $V= \reals\times \reals^d$, $V^+= 
 \set{(t, \vtr x)| t\ge \norm{\vtr x}}$, and $e= (1, \vtr 0)$, where 
 $\norm{\vtr x}$ is a norm in $\reals^d$.
Such cones only allow dichotomic observables in the sense that $e-\exbound V^+= 
 \exbound V^+$.
However several interesting cases are instances of this example:
the event space of tossing a coin (classical bit, $d= 1$ and $\norm{\vtr x}= 
 \abs{\vtr x_1}$),
the local part of a Popescu-Rohrlich box \cite{Popescu:1994FPH} (generalized 
 bit \cite{Barrett:2007PRA}, $d= 2$ and $\norm{\vtr x}= \abs{\vtr x_1} + 
 \abs{\vtr x_2}$),
the quantum mechanical two-level system (quantum bit, $d= 3$ and $\norm{\vtr 
 x}= \sqrt{\vtr x\cdot \vtr x}$),
and ``hyperbits'' \cite{Pawlowski:2012PRA} which generalize the quantum bit by 
 allowing for $d>3$ while keeping the Euclidean norm.
The states for a dichotomic norm cone are the maps $(t,\vtr x)\mapsto t+\vtr 
 w\cdot \vtr x$ with $\norm{\vtr w}_*\le 1$, where $\norm{\vtr w}_*\equiv 
 \sup\set{\vtr w\cdot \vtr y | \norm{\vtr y}\le 1}$ is the dual norm.
The order norm is also easy to evaluate, $\norm{(t,\vtr x)}= \abs{t} + 
 \norm{\vtr x}$.

\vspace{.5em}\noindent\textbf{A pathological example.}
We define $V^+= \cone \set{a_1,a_2,\dotsc,a_6}$ where $a_1,\dotsc,a_4$ is a 
 basis of $V$, $a_5= a_1-a_3+a_4$, and $a_6= a_2+a_3-a_4$.
The order unit is chosen to be $e= a_1+a_2+\frac 12(a_3+a_4)$.
This case is pathological in the sense that there is no way to write $e= 
 \sum_{v\in \mathcal A} v$ for any $\mathcal A\subset \set{a_1,\dotsc, a_6}= 
 \exbound V^+$.

\section{Sequential measurements}\label{s8520}
We now discuss sequential measurements for such generalized probabilistic 
 models for which the measurement effects can be squeezed into an AOU space 
 $(V,V^+,e)$.
That is, any measurement can be described by a family of effects $(f_k)\subset 
 V^+_e$ with $\sum_k f_k= e$---this is in analogy to the generalized 
 measurements that occur in quantum mechanics.
Following the discussion in Sec.~\ref{s1003}, we consider the situation that a 
 sequence of two measurements has been performed and the consecutive outcomes 
 $f,g\in V^+_e$ have occurred.
What is the prediction for the probability $\probability_\omega(f\then g)$ for 
 the event $f\then g$, given that the system was in a state $\omega\in \mathcal 
 S$?

This probability will clearly depend on the actual implementation of the first 
 measurement and this implementation is readily summarized by a map $\phi\colon 
 V\rightarrow V$, so that $\probability_\omega(f\then g)= \omega[\phi(g)]$.
This implies that $\phi$ is positive and for consistency we assume $\phi(e)= 
 f$, i.e., the all-embracing outcome $e$ occurs with unit probability, given 
 that previously the outcome $f$ has occurred.
We also assumed that $\phi$ is linear, so that performing with probability $p$ 
 a measurement with outcome $g$ and with probability $1-p$ a measurement with 
 outcome $h$ obeys $\probability[f\then pg+ (1- p)h)]= p\probability(f\then 
 g) + (1- p)\probability(f\then h)$.
A positive map $\phi$ with $\phi(e)= f$ is called $f$-compatible 
 \cite{Heinosaari:2012}.

In principle, any choice of an $f$-compatible map\footnote{%
In quantum mechanics we would be restricted to completely positive maps, but 
 this subtlety can be ignored for the discussion here.}
may be suitable to describe $f\then g$.
Here we are concerned about the projective measurements which generalize 
 Lüders's rule.
The following notions capture important properties of Lüders's rule.
\begin{definition}
Let $\phi$ be an $f$-compatible map for $f\in V^+_e$, i.e., $\phi(e)= f$ and 
 $\phi(V^+)\subset V^+$.
\begin{itemize}\itemsep-.1em
\item[(i)]
$\phi$ is projective, if $\phi\circ\phi= \phi$.
\item[(ii)]
$\phi$ is neutral, if $\omega\circ \phi= \omega$ for any $\omega\in \mathcal S$ 
 with $\omega(f)= 1$.
\item[(iii)]
$\phi$ is coherent, if $\phi(g)= g$ for any $g\in V^+$ with $g\le f$.
\end{itemize}
\end{definition}

One might be tempted to use $f$-compatible projections for defining a 
 generalization of Lüders's rule.
For an extremal element, $f\in \exbound V^+$, such a map is of the form $\phi= 
 f\omega$, where $\omega\in \mathcal S$ is a state with $\omega(f)= 1$ [the 
 existence of such a state is due to Proposition~\ref{p10788}~(ii)].
In quantum mechanics this already yields uniquely Lüders's rule for rank-one 
 projections.
Furthermore, any family $(f_k)\subset V_e^+$ with $\sum f_k\le e$ and 
 $f_k$-compatible projections $\phi_k$ enjoys perfect repeatability, 
 $\phi_k\circ \phi_\ell= \delta_{k, \ell} \phi_k$, utilizing the Kronecker 
 symbol $\delta_{k, \ell}$.
This holds, since for $k\ne \ell$ and any $h\in V_e^+$ we have $0\le 
 \phi_k\phi_\ell h\le \phi_k \phi_\ell e = \phi_k f_\ell = - 
 \phi_k(e-f_k-f_\ell)\le 0$.

Unfortunately, projectivity does not sufficiently fix the choices for $\phi$.
For example, $\phi= e\omega$ is an $e$-compatible projection, but any 
 subsequent measurement will solely depend on the arbitrary choice of 
 $\omega\in \mathcal S$.
Previously \cite{Mielnik:1969CMP, Araki:1980CMP, Niestegge:2008FPH, 
 Ududec:2011FPH}, filters have been considered as a possible extensions of 
 Lüders's rule to generalized probabilistic models.
A filter is a neutral $f$-compatible projection, but it is only called a filter 
 if there also exists a neutral $f$-compatible projection for $e-f$.
Here, we study a different extension of Lüders's rule, namely the coherent 
 Lüders's rules.
\begin{definition}\label{d29446}
A coherent Lüders's rule (CLR) for $f\in V_e^+$ is a coherent $f$-compatible 
 map.
\end{definition}
\noindent
We occasionally write $f^\sharp$ for a CLR of $f$, although this map is not 
 necessarily uniquely defined by the above condition.

A possible interpretation behind the definition of coherence is that the 
 relation $g\le f$ indicates that the outcome $g$ provides always a finer 
 information than $f$ in the sense that independent of the state $\omega$ of 
 the system, $g$ is always less likely to be triggered than $f$.
Thus getting firstly the course grained information $f$ and then the fine 
 grained information $g$ is assumed not to influence $g$.
Hence $f$ preserves all the ``coherences'' of $g$.
We also refer to Proposition~\ref{p20023}, Proposition~\ref{p10596}, the 
 example of a triple-slit experiment in Sec.~\ref{s26365}, and the Discussion 
 in Sec.~\ref{s12113} for further reasoning in favor of this definition.
In Sec.~\ref{s1894} it is also shown that neutral $f$-compatible projections 
 and coherent $f$-compatible maps are different concepts.

\subsection{Basic properties of coherent Lüders's rules}
There are several equivalent ways of expressing Definition~\ref{d29446}.
\begin{lemma}\label{l17842}
For a positive map $\phi$ and an effect $f\in V_e^+$, the following statements 
are equivalent.
\begin{itemize}\itemsep-.1em
\item[(i)]
$\phi(e)= f$ and $\phi(g)= g$ for all $0\le g\le f$.
\item[(ii)]
$\phi(e)\le f$ and $\phi(g)\ge g$ for all $0\le g\le f$.
\item[(iii)]
$a\le \phi(g)\le f\norm g$ for all $g\in V^+$, whenever $0\le a\le f$ and 
 $a\le g$.
\item[(iv)]
$a\le \phi(g)\le f$ for all $g\in V_e^+$, whenever $0\le a\le f$ and $a\le g$.
\end{itemize}
\end{lemma}
\begin{proof}
In order to see that (i) implies (iii), note that $\phi(g)= \phi(g-a)+ a\ge a$.
Furthermore, $f\norm g-\phi(g)\ge 0$ follows immediately when considering
 $\phi(\norm{g}e- g)\ge 0$ and by fact that $\norm g e\ge g$ holds since $e$ 
 is Archimedean.

Obviously (iii) implies (iv), since for $g\in V_e^+$ we have $\norm g\le 1$.

Statement (ii) follows from (iv) by letting $g_\text{(iv)}= e$ (yielding 
 $\phi(e)\le f$) and by choosing $g_\text{(iv)}= g_\text{(ii)}= a$ (yielding 
 $\phi(g_\text{(ii)})\ge g_\text{(ii)}$).

We finally show that (i) follows from (ii).
We first use that $\phi(e-f)\ge 0$ and thus $f\ge \phi(e)\ge \phi(f)\ge f$, 
 i.e., $\phi(e)= f= \phi(f)$.
Then $\phi(g)-g\le \phi(f)-f\equiv 0$, where the inequality follows from 
 $f-g\le \phi(f-g)$, which is due to $0\le f-g\le f$.
But $\phi(g)\le g$ can only be compatible with $\phi(g) \ge g$ when 
 $\phi(g)=g$.
\end{proof}
\noindent
Note, that with statement (iv) of this lemma, we have $\phi(h)=f$ for $f\le h 
 \le e$, by letting $a=f$ and $g=h$.

From a physical perspective, a CLR for $f$ describes exactly such a measurement 
 that does not disturb any other subsequent measurement with outcome $f$.
\begin{proposition}\label{p20023}
Let $\mathcal C\supset (V^+\otimes \mathcal S)$ be some cone of positive maps 
 and let $\phi$ be an $f$-compatible map for $f\in V^+_e$.
Then $\phi$ is coherent if and only if $\phi\circ \psi=\psi$ holds for all 
 $f$-compatible maps $\psi\in \mathcal C$.
\end{proposition}
\begin{proof}
If $\psi$ is $f$-compatible, then $\psi(h)\le \psi(e)= f$ for any $h\in V^+_e$.
It follows that $\phi\circ \psi=\psi$ if $\phi$ is a CLR.
For the converse we consider $\psi=(f-g)\omega + g\sigma\in \mathcal C$ with 
 $0\le g\le f$ and $\omega,\sigma\in \mathcal S$.
This map is clearly $f$-compatible and we define $\Delta\equiv 
 \phi\circ\psi-\psi= [\phi(f)-f]\omega + [\phi(g)-g](\sigma-\omega)$.
From $\Delta(e)=0$ we obtain $\phi(f)=f$ and assuming $\sigma\ne \omega$, also 
 $\phi(g)=g$ must hold.
Hence $\phi$ is coherent.
\end{proof}

A CLR in particular obeys repeatability and compatibility.
\begin{proposition}\label{p10596}
Let $f^\sharp$ and $g^\sharp$ be two CLRs for $f,g\in V_e^+$, respectively.
We have:
\begin{itemize}\itemsep-.1em
\item[(i)] $f^\sharp$ is projective.
\item[(ii)] If $g\le f$ then $f^\sharp g= g^\sharp f$.
\item[(iii)]
If $g\le f$ and $g^\sharp$ is unique for $g$, then $f^\sharp g^\sharp= g^\sharp 
 f^\sharp$.
\end{itemize}
\end{proposition}
\begin{proof}
We implicitly use Lemma~\ref{l17842}~(iv).
Then $f^\sharp h\le f$ for any $h\in V_e^+$ and hence $f^\sharp (f^\sharp h)= 
 f^\sharp h$.
If $g\le f$ then immediately $f^\sharp g= g= g^\sharp f$ (cf.\ also the remark 
 after Lemma~\ref{l17842}).
If the CLR for $g$ is unique then $f^\sharp g^\sharp= g^\sharp f^\sharp$, since 
 $f^\sharp g^\sharp= g^\sharp$ and on the other hand $g^\sharp f^\sharp$ is a 
 valid CLR for $g$.
\end{proof}

We mention that the property of being neutral or coherent is robust under 
 sections.
A section \cite{Kleinmann:2013PRL} is a positive unital injection $\tau$ from 
 $(W,W^+,e')$ to $(V,V^+,e)$, such that there exists a positive surjection 
 $\tau'\colon V\rightarrow W$ with $\tau'\circ\tau= \id_W$.
If $\phi$ is a neutral/coherent $\tau(f)$-compatible map, then $\tau'\circ 
 \phi\circ \tau$ is a neutral/coherent $f$-compatible map.
An important instance of this observation is the embedding of the classical 
 events into quantum events via $\tau \colon \vtr v\mapsto \diag(\vtr v)$.
In contrast, general $\tau(f)$-compatible projections do not always induce 
 $f$-compatible projections.

\subsection{Conditions on elements with a coherent Lüders's rule}\label{s20245}
Not all $f\in V_e^+$ admit a CLR as we see next.
But the CLR for $e$ is the identity mapping, while for $0$ it is the zero 
 mapping.
On the other hand, if $f$ is extremal, $f\in \exbound V^+$, then any 
 $f$-compatible projection is a CLR.
For the general situation we have
\begin{proposition}\label{p14167}
For $f\in V_e^+$ consider the following statements.
\begin{itemize}\itemsep-.1em
\item[(i)]
$f$ admits a CLR.
\item[(ii)]
$g\le f\norm g$ for all $0\le g\le f$.
\item[(iii)]
$0\le g\le f$ and $g\le e-f$ only for $g= 0$.
\end{itemize}

Then (i) implies (ii) and (ii) implies (iii).
\end{proposition}
\begin{proof}
Statement (ii) is a direct consequence of Lemma~\ref{l17842}~(iii), $g= 
 f^\sharp g\le f \norm g$.
For the second part we consider $0\le g\le f\le e-g$.
Then $0\le g\le f\norm g \le \norm g (e-g)$ and therefore $e \norm g/(\norm 
 g+1) \ge g$, which contradicts $\norm g\equiv \inf\set{ r\in \reals^+ | r e 
 \ge g}$ unless $\norm g= 0$.
By the Archimedean property the assertion follows.
\end{proof}

From part (ii) of this proposition it immediately follows that if $f= \sum_k 
 p_k f_k$ with $(f_k)\subset \exbound V^+$ and real numbers $p_k> 0$ then 
 already $f_k\le f$.
But one cannot conclude that there exists a decomposition of $f$ into extremal 
 elements with unit weights, cf.\ the pathological example form 
 Sec.~\ref{s9680} with $f=e$.
This pathological space also provides an example where (iii) does not imply 
 (ii).
The counterexample works with $f= e-a_1-a_2\equiv (a_3+a_4)/2$, which obeys 
 (iii).
But $f-p a_3\ge 0$ only for $p\le \frac 12$ in contradiction to (ii).
At the moment it remains unclear whether (ii) implies (i), even though it does 
 not seem plausible to hold.
On the other hand, for quantum mechanics, already statement (iii) can only hold 
 if $F$ is a projection since $0\le \sqrt F(\openone -F)\sqrt F \equiv F -F^2 
 \le F$ and $0\le (\openone-F)^2\equiv \openone- 2F +F^2$, i.e., $F-F^2 \le 
 \openone -F$.
By assumption we then have $F-F^2= 0$ and hence $F$ is a projection.

\subsection{Neutral maps}\label{s1894}
Neutral $f$-compatible projections have been suggested previously 
 \cite{Mielnik:1969CMP, Araki:1980CMP, Niestegge:2008FPH, Ududec:2011FPH} as an 
 extension of Lüders's rule to generalized probabilistic models.
For the moment we call them neutral Lüders's rules (NLRs).
If $f$ and $e-f$ allow an NLR, then an NLR for $f$ is a filter.
We observe:

\vspace{.5em}\noindent
\textit{1.\ Some elements do not have an NLR, despite being extremal.}
Consider the dichotomic norm cone (cf.\ Sec.~\ref{s9680}) with $\norm{\vtr x}= 
 \sum \abs{\vtr x_i}$ and $d\ge 2$.
In this case, there exists no neutral map $\phi$ for any of the extremal 
 elements $f\in \exbound V^+$ since states with $\omega(f)= 1$ are not unique 
 but on the other hand $\phi= f\omega$ must hold for $\phi$ to be an 
 $f$-compatible projection.

\vspace{.5em}\noindent
\textit{2.\ Some elements with an NLR do not have a CLR.}
An example occurs in the pathological example from Sec.~\ref{s9680} for the 
 effect $f= e-a_1-a_2$.
As demonstrated at the end of Sec.~\ref{s20245} this element does not have a 
 CLR.
But the only state with $\omega(f)= 1$ is $\omega(a_k)= (0,0,1,1,0,0)_k$ and 
 hence $f\omega$ is an NLR for $f$.
One can also construct an NLR for the complement $f_\neg= e-f$, showing that 
 $f\omega$ is a filter.
The NLR for $f_\neg$ is not unique, but a possible representative is given by 
 $a_1 \omega_1 + a_2 \omega_2$ with $\omega_i(a_k)= 
 \delta_{i,k}+\delta_{i+4,k}$.

\section{Applications}
\subsection{Quantum mechanics}\label{s15325}
In quantum mechanics, $F\in V^+_e$ admits a CLR if and only if it is a 
 projection.
We have shown necessity in Sec.~\ref{s20245} and in order to show sufficiency 
 we assume that $F$ is a projection and that $F^\sharp(X)= FXF$.
It remains to show that $G= FGF$ for any $0\le G \le F$.
Although this is an easy and well-known relation, we shall spend a few lines to 
 show it:
We write $F_\neg= \openone-F$.
Then $0\le F_\neg(F-G)F_\neg= - F_\neg G F_\neg \le 0$ and thus $F_\neg G= 
F_\neg GF$.
But $0\le (F+\lambda F_\neg)G(F+\lambda F_\neg)= FGF + \lambda(F_\neg GF+ FG 
 F_\neg)$ for all $\lambda\in \reals$ implies $F_\neg G F= -F GF_\neg$, i.e., 
 $G= FGF$.

The rule $F^\sharp\colon X\mapsto FXF$ is unique as we demonstrate by 
 construction.
Assume $G\in V_e^+$.
Then $0\le F (\openone -G) F = F-FGF$ implying $F^\sharp(FGF)= FGF$ and $0\le 
 F^\sharp[F_\neg(\openone - G)F_\neg] = - F^\sharp(F_\neg G F_\neg)$ which 
 yields $F^\sharp(F_\neg G F_\neg)= 0$.
With $G'_\lambda\equiv (F+ \lambda F_\neg) G (F+ \lambda F_\neg)\ge 0$ we have
\begin{equation}
 F^\sharp(G'_\lambda)= FGF + \lambda A\ge 0
 \text{, where } A= F^\sharp(F_\neg GF+ FG F_\neg),
\end{equation}
 for all $\lambda\in \reals$.
This implies again $A= 0$ and hence $F^\sharp(G)\equiv F^\sharp(G'_1)= FGF$.

We mention that we did not assume that $F^\sharp$ is completely positive but 
 nevertheless obtained the intended quantum mechanical Lüders's rule.

\subsection{Dichotomic norm cones}\label{s16520}
As a second example, we consider the dichotomic norm cones of Sec.~\ref{s9680}.
For this AOU spaces the set of effects admitting a CLR is given by 
 $\set{0,e}\cup \exbound V^+$, cf.\ Appendix~\ref{a25858}.
This shows that dichotomic norm cones form a very convenient toy model for 
 which basically the assumption of an $f$-compatible projection alone leads to 
 a reasonable Lüders's rule.
Put into an explicit form, any extremal element $f\in \exbound V^+$ is of the 
 form $f= (\tfrac 12,\vtr f)$ with $\norm{\vtr f}= \tfrac 12$ and any 
 corresponding CLR reads thus
\begin{equation}\label{e26441}
 f^\sharp\colon (t,\vtr x)\mapsto (t+\vtr f'\cdot \vtr x) f
 \text{, with } \vtr f'\cdot \vtr f= \tfrac12
 \text{, and } \norm{\vtr f'}_*= 1.
\end{equation}
Since the set of CLRs for a given effect $f$ is convex, it follows that if 
 $\norm{\vtr x}$ is a $p$-norm with $1<p<\infty$ then the CLR is unique.
This is due to the fact that then the dual norm $\norm{\vtr x}_*$ is the 
 $[p/(p-1)]$-norm, the unit-sphere of which only has convex subsets with a 
 single vector.
On the other hand, for the Manhattan Norm, $p= 1$, and e.g.\ $\vtr f= 
 (\frac12,0,\dotsc,0)$ the available choices are any of $\vtr f'= 
 (1,\xi_2,\dotsc,\xi_d)$ with arbitrary coefficients $-1\le \xi_k\le 1$.

As an example we compute the effective ``observable`` for an sequential 
 measurement of two dichotomic observables $A= a-a_\neg$ and $B= b-b_\neg$ with 
 $a_\neg= e-a$ and $b_\neg= e-b$.
That is, with the notation $A\sharp= a^\sharp-a_\neg^\sharp$, we aim at 
$A\sharp B$.
For simplicity we assume that in $a^\sharp$ and $a_\neg^\sharp$ we have $\vtr 
 a_\neg'= -\vtr a'$, which surely holds when both CLRs are unique.
Writing $b= (\beta,\vtr b)$ yields
\begin{equation}\label{e22968}
 A\sharp B= (2\beta -1)A+ 2(\vtr a'\cdot \vtr b) e.
\end{equation}
If $\beta= \frac12$, e.g., because $b$ is extremal, then the expected value 
 $\exv{A\sharp B}_\omega\equiv \omega(A\sharp B)$ does not depend on the 
 prepared state $\omega$.
For the case of the Euclidean norm, $\norm{\vtr x}= \sqrt{\vtr x\cdot \vtr x}$, 
 and $B\sharp$ defined analogously to $A\sharp$, we find in addition $A\sharp 
 B= B\sharp A$.
Both aspects have been observed already for qubits \cite{Fritz:2010JMP} which 
 corresponds to the dichotomic norm cone with $d= 3$ and the Euclidean norm.

~

Dichotomic norm cones can have strong non-quantum behavior.
As an example we consider the simplest correlation term $\exv{\text{LG}'}$,
\begin{equation}\label{e5277}
\exv{\text{LG}'}_\omega= \omega(A\sharp B+ B- A).
\end{equation}
For so-called macro-realistic systems (which are in our language CLR 
 measurements on the classical events) the constraint $\exv{\text{LG}'}\le 1$ 
 is valid \cite{Leggett:1985PRL}, while for quantum mechanics the bound 
 $\exv{\text{LG}'}\le \frac32$ is in order \cite{Budroni:2013PRL}.
Note, that the quantum mechanical bound only holds for CLRs 
 \cite{Budroni:2014PRL}.
For dichotomic norm cones and assuming again that always $\vtr a'= -\vtr 
 a_\neg'$ we obtain the sharp bound (cf.\ Appendix~\ref{a26355})
\begin{equation}\label{e32092}
 \exv{\text{LG}'}\le 2\norm{\vtr b- \vtr a} + 2\vtr a'\cdot \vtr b
  \text{, where } \norm{\vtr b}=\tfrac 12.
\end{equation}
In the case of the Manhattan norm, $\norm{\vtr x}= \sum \abs{\vtr x_k}$, and 
 $d= 2$ we find that the r.h.s.\ of this inequality can easily reach $3$ by 
 choosing $\vtr a= (\tfrac 12,0)$, $\vtr b= (0,\tfrac 12)$, and $\vtr a'= 
 (1,1)$.

~

We finally mention that Spekkens's toy model \cite{Spekkens:2007PRA} implements 
 a CLR.
In this model, there are six extremal elements $\exbound V^+= \set{a_{\pm1}, 
 a_{\pm2}, a_{\pm3}}$ given by $a_i= (\tfrac12,\vtr a_{(i)})$, with
\begin{equation}
 \vtr a_{(\pm1)}= (\pm \tfrac12,0,0) \text{, }
 \vtr a_{(\pm2)}= (0,\pm \tfrac12,0) \text{, and }
 \vtr a_{(\pm3)}= (0,0,\pm \tfrac12).
\end{equation}
These elements form observables $A_k= a_{+k}-a_{-k}$ and hence $e= 
 a_{+k}+a_{-k}\equiv (1,\vtr 0)$.
This way Spekkens's toy model is the dichotomic norm cone with $d= 3$ and the 
 Manhattan norm.
Spekkens also introduced a state update rule for this model, which is such that 
\begin{equation}
 \probability(a_i\then a_j)= \probability(a_i)\begin{cases}
  1 & i= j, \\
  0 & i= -j, \\
 \tfrac12 & \text{else.}
 \end{cases}
\end{equation}
This update rule corresponds to the CLR defined in Eq.~\eqref{e26441} with the 
 choice $\vtr a_{(i)}'= 2\vtr a_{(i)}$.

\subsection{The triple-slit experiment}\label{s26365}
While the double-slit experiment is a prime example of a quantum effect, within 
 quantum mechanics there are no higher order interference terms, as has been 
 found by Sorkin \cite{Sorkin:1994MPLA}.
This absence was also verified in experiments \cite{Sinha:2010SCI}.
Recently, the triple-slit experiment has been investigated as instance of 
 sequential measurements in the context of generalized probabilistic models 
 \cite{Ududec:2011FPH} and the (im)possibility of triple-slit correlations in 
 such models was discussed e.g.\ in Refs.~\onlinecite{Niestegge:2012AMP,
 Dakic:2014NJP}.

In an $n$-slit experiment with slits labeled by $\mathcal N= \set{1, 2, \dotsc, 
 n}$, detecting that the particle passed through any of the slits 
 $\alpha\subset \mathcal N$ plays the role of the first measurement, described 
 by a map $\phi_\alpha$.
The measurement of the interference pattern on the screen is hence the second 
 measurement.
Each possible combination of open slits $\alpha$ may have its particular 
 interference pattern as long as the integrated intensity is independent of 
 whether the slits are opened individually or jointly, so that $\phi_\alpha(e)= 
 \sum_{k\in \alpha}\phi_{\set k}(e)$.
Clearly, the total intensity is bounded by unity, so that $\phi_{\mathcal 
 N}(e)\in V^+_e$.

We discuss now briefly the assumption that $\phi_\alpha$ is coherent for the 
 effect $\phi_\alpha(e)$ and hence is a CLR.
Assume that the probability for an effect $g$ depends only on the integrated 
 intensity that arrives through the slits $\alpha$, i.e., $\phi_\alpha(e)\equiv 
 \sum_{k\in \alpha}\phi_{\set k}(e)\ge g$.
In this case, the coherence assumption $\phi_\alpha(g)= g$ assures that putting 
 the simultaneously opened slits $\alpha$ in front of a measurement with 
 outcome $g$ does not change that outcome.

We recursively define (in general non-positive) maps $\eta_\alpha$ via
\begin{equation}\label{e3498}
 \phi_\alpha= \sum_{\beta\subset \alpha}\eta_\beta.
\end{equation}
Then those maps $\eta_\alpha$ are exactly the interference terms 
 $I_{\abs\alpha}(\alpha)$ as defined by Sorkin \cite{Sorkin:1994MPLA}, adapted 
 to the language chosen here.
In Eq.~\eqref{e3498} we try to write the map on the l.h.s.\ in terms of the 
 lower order correlations.
The difference between the actual map $\phi_\alpha$ and this lower order sum is 
 then defined as $\eta_\alpha$.

In a quantum mechanical $n$-slit experiment the slits are described by 
 projections $\Pi_k$ obeying $\sum\Pi_k\le \openone$.
We let $\Pi_\alpha= \sum_{k\in \alpha}\Pi_k$ and therefore
\begin{equation}
 \phi_\alpha\colon X\mapsto \Pi_\alpha X \Pi_\alpha
 \equiv \sum_{\beta \subset\alpha\colon \abs{\beta}\le 2} \eta_\beta(X),
\end{equation}
 that is, $\eta_\beta= 0$ whenever $\abs\beta>2$.
That is, in quantum mechanics all interference terms above the second order 
 vanish.
We mention that in general this absence only occurs if the quantum instrument 
 follows Lüders's rule, as a counterexample may serve $\phi_\alpha\colon 
 X\mapsto \sqrt{A_\alpha}X\sqrt{A_\alpha}$ with $A_\alpha= \sum_k A_k$ and 
 $A_1= \openone/2$, $A_2= \proj0/2$, $A_3= \proj1/2$.
Such measurements, however, may fail to have a proper physical interpretation 
 as a triple-slit experiment, since the operators $A_k$ may act non-locally.

For generalized probabilistic models, though, we can easily have higher order 
 correlations:
Consider the AOU space with $V^+= \cone \set{a_1,\dotsc,a_5}$, where 
 $a_1,\dotsc,a_4$ is a basis of $V$, $a_5= a_1+a_2+a_3-a_4$, and $e= 
 a_1+a_2+a_3\equiv a_4+a_5$.
We choose $\phi_{\alpha}(e)= \sum_{k\in\alpha} a_k$ for 
 $\alpha\subset\set{1,2,3}\equiv \mathcal N$.
A brief calculation yields for $\alpha\subsetneq \mathcal N$,
\begin{equation}\label{e29484}
 \phi_\alpha= \sum_{k\in \alpha}a_k\omega^\alpha_k
\end{equation}
 where $\omega^\alpha_k$ are arbitrary choices of states with 
 $\omega^\alpha_k(a_k)= 1$.
Since those states are not unique, we can e.g.\ use this freedom to achieve 
 commutativity, $\phi_\alpha\circ \phi_\beta= \phi_\beta\circ \phi_\alpha$, or 
 to get vanishing double-slit correlations, $\eta_{\set{k, \ell}}= 0$.
In contrast, the map for the triple-slit is the identity mapping, 
 $\phi_{\mathcal N}= e^\sharp\equiv \id$.
From Eq.~\eqref{e29484} we see that $a_4\notin \eta_\alpha(V)$ except for 
 $\alpha= \mathcal N$, i.e., nonvanishing triple-slit correlations occur.

\section{Discussion}\label{s12113}
An important property of quantum systems is that the measurement necessarily 
 changes the state of the system---or in a Heisenberg type-of-picture that the 
 description of a measurement depends on previous measurements that have been 
 performed.
How this change occurs in general depends on the actual implementation of the 
 measurement.
In quantum mechanics, however, the change induced by projective measurements 
 according to Lüders is the least disturbing and least biased implementation of 
 a projective measurement.
We re-derived this rule in quantum mechanics (cf.\ Sec.~\ref{s15325}) solely 
 from the coherence assumption stated in Definition~\ref{d29446}.
This definition of coherent Lüders rules (CLRs) can be applied to a wide class 
 of hypothetical non-quantum models, namely the generalized probabilistic 
 models which can described by means of Archimedean ordered vector spaces.

We showed in Proposition~\ref{p20023} that CLRs are exactly those maps which do 
 not disturb any subsequent and possibly more ``noisy'' implementation of the 
 same measurement.
We also showed that familiar results of repeatability and compatibility hold 
 (Proposition~\ref{p10596}, cf.\ also Refs.~\onlinecite{Mielnik:1969CMP, 
 Niestegge:2008FPH}).

In quantum mechanics, Lüders's rule is directly linked to and singles out the 
 projection operators, which in turn play a key role e.g.\ in spectral theory.
(Celebrated results for a generalized spectral theory \cite{Alfsen:1976MAMS, 
 Araki:1980CMP, Alfsen:2003} are, however, linked to neutral maps.)
We find that for extremal measurement effects (a generalization of rank-one 
 projections in quantum mechanics) an CLR always exists, while necessary 
 conditions for existence have been given in Proposition~\ref{p14167}.
Also, in certain pathological cases, the CLR is not unique.
This ambiguity might be unsatisfactory, but for quantum mechanics and classical 
 mechanics the conditions of being a CLR are sufficient to achieve uniqueness, 
 so that adding any further condition is of a rather speculative kind.

Finally we demonstrated in Sec.~\ref{s16520} that CLRs occurred already earlier 
 in Spekkens's toy model \cite{Spekkens:2007PRA} and that this toy model can 
 now be seen as an instance of a much wider class of models with a natural 
 notion of sequential measurements.
For those models it is e.g.\ straightforward to compute the upper limit for the 
 Leggett-Garg inequality in Eq.~\eqref{e32092}.
As a last instance we discussed in Sec~\ref{s26365} the triple-slit experiment, 
 finding that generalized probabilistic models with a CLR can easily have 
 substantial triple-slit correlations, while it is an important prediction of 
 quantum mechanics that those are absent.

\begin{acknowledgments}
The proof of Lemma~\ref{l17842} was simplified by one of the anonymous 
 referees.
For discussions, hints, and amendments I am particularly indebted to
J.\ Emerson,
O.\ Gühne,
R.\ Hübener,
J.-Å.\ Larsson,
V.B.\ Scholz,
M.\ Ziman, and
Z.\ Zimborás.
I thank the Centro de Ciencias de Benasque, where part of this work has been 
done, for its hospitality during the workshop on quantum information
2013.
I acknowledge support from
the BMBF (Chist-Era Project QUASAR),
the Brazilian agency CAPES, through the program Science without Borders,
the DFG,
the EU (Marie Curie CIG 293993/ENFOQI), and
the FQXi Fund (Silicon Valley Community Foundation).
\end{acknowledgments}

\appendix
\section{Elements with a coherent Lüders's rule in dichotomic norm 
cones}\label{a25858}
In a dichotomic norm cone (cf.\ Sec.~\ref{s9680}), the set of effects admitting 
 a CLR is given by $\set{0,e}\cup \exbound V^+$, as stated in 
 Sec.~\ref{s16520}.
For $f= (t,\vtr f)\in V^+_e$ we have $\norm f= 1$ if and only if $t= 1- 
 \norm{\vtr f}$ and $\norm{\vtr f}\le \frac12$.
Assume now that $f$ admits a CLR, but $0\ne f\ne e$.
By virtue of Proposition~\ref{p14167}~(ii) it follows that
 $\norm f= 1$ and $\norm{\vtr f}= \frac12$.
The first statement is obtained by choosing $g= f$ and the second statement by 
 the choice $0\le g= (1-2\norm{\vtr f})e= f-(\norm{\vtr f},\vtr f)\le f$.
If now $a\in \exbound V^+$ and $p>0$, such that $p a\le f$, then also $a\le f$.
This reads $\frac12- \frac12\ge \norm{\vtr f- \vtr a}$ and therefore $f= a$.

\section{Obtaining Eq.~\eqref{e32092}}\label{a26355}
Under the result $A\sharp B= (2\beta-1)A+2(\vtr a'\cdot \vtr b)e$ 
 [Eq.~\eqref{e22968}] we bound the correlation term $\exv{\text{LG}'}_\omega= 
 \omega(A\sharp B+ B- A)$ [Eq.~\eqref{e5277}] for dichotomic norm cones, 
 assuming $A= a- a_\neg= (0,2\vtr a)$, and $B= b- b_\neg= (2\beta-1,2\vtr b)$.
Writing $\omega= (1,\vtr w)$, this yields for $\vtr b\ne \vtr 0$,
\begin{equation}\begin{split}
\tfrac12 \exv{\text{LG}'}_\omega&=
 \vtr a'\cdot \vtr b+ \vtr w\cdot(\vtr b-\vtr a) +
  (2\beta-1)(\vtr w\cdot \vtr a+\tfrac 12)\\
 &\le \norm{\vtr b} [\vtr a'\cdot \underline{\vtr b}+
  \norm{\underline{\vtr b}- 2\vtr a}-1]+ \tfrac12\\
\end{split}\end{equation}
 with $\underline{\vtr b}= \vtr b/\norm{\vtr b}$.
The inequality is due to $\beta\le 1-\norm{\vtr b}$,
 $\norm{\vtr w}_*\le 1$, and $\vtr w\cdot\vtr a\ge - \tfrac 12$.
The bound is sharp, if $\beta= 1-\norm{\vtr b}$ and $\vtr w\cdot 
 (\underline{\vtr b}-2\vtr a)= \norm{\underline{\vtr b}-2\vtr a}$.
Using the conditions from Eq.~\eqref{e26441}, we have $\norm{\underline{\vtr 
 b}- 2\vtr a}\ge -\vtr a'\cdot (\underline{\vtr b}- 2\vtr a)= 1- \vtr a'\cdot 
 \underline{\vtr b}$ and hence the term in square brackets is never negative.
This makes the choice $\norm{\vtr b}= \tfrac12$ optimal and we arrive at the 
 sharp bound of Eq.~\eqref{e32092}.

\bibliography{the,xxx}

\begin{thebibliography}{34}%
\makeatletter
\providecommand \@ifxundefined [1]{%
 \@ifx{#1\undefined}
}%
\providecommand \@ifnum [1]{%
 \ifnum #1\expandafter \@firstoftwo
 \else \expandafter \@secondoftwo
 \fi
}%
\providecommand \@ifx [1]{%
 \ifx #1\expandafter \@firstoftwo
 \else \expandafter \@secondoftwo
 \fi
}%
\providecommand \natexlab [1]{#1}%
\providecommand \enquote  [1]{``#1''}%
\providecommand \bibnamefont  [1]{#1}%
\providecommand \bibfnamefont [1]{#1}%
\providecommand \citenamefont [1]{#1}%
\providecommand \href@noop [0]{\@secondoftwo}%
\providecommand \href [0]{\begingroup \@sanitize@url \@href}%
\providecommand \@href[1]{\@@startlink{#1}\@@href}%
\providecommand \@@href[1]{\endgroup#1\@@endlink}%
\providecommand \@sanitize@url [0]{\catcode `\\12\catcode `\$12\catcode
  `\&12\catcode `\#12\catcode `\^12\catcode `\_12\catcode `\%12\relax}%
\providecommand \@@startlink[1]{}%
\providecommand \@@endlink[0]{}%
\providecommand \url  [0]{\begingroup\@sanitize@url \@url }%
\providecommand \@url [1]{\endgroup\@href {#1}{\urlprefix }}%
\providecommand \urlprefix  [0]{URL }%
\providecommand \Eprint [0]{\href }%
\providecommand \doibase [0]{http://dx.doi.org/}%
\providecommand \selectlanguage [0]{\@gobble}%
\providecommand \bibinfo  [0]{\@secondoftwo}%
\providecommand \bibfield  [0]{\@secondoftwo}%
\providecommand \translation [1]{[#1]}%
\providecommand \BibitemOpen [0]{}%
\providecommand \bibitemStop [0]{}%
\providecommand \bibitemNoStop [0]{.\EOS\space}%
\providecommand \EOS [0]{\spacefactor3000\relax}%
\providecommand \BibitemShut  [1]{\csname bibitem#1\endcsname}%
\let\auto@bib@innerbib\@empty
\bibitem [{\citenamefont {Spekkens}(2007)}]{Spekkens:2007PRA}%
  \BibitemOpen
  \bibfield  {author} {\bibinfo {author} {\bibfnamefont {R.~W.}\ \bibnamefont
  {Spekkens}},\ }\bibfield  {title} {\enquote {\bibinfo {title} {Evidence for
  the epistemic view of quantum states: A toy theory},}\ }\href {\doibase
  10.1103/PhysRevA.75.032110} {\bibfield  {journal} {\bibinfo  {journal} {Phys.
  Rev. A}\ }\textbf {\bibinfo {volume} {75}},\ \bibinfo {pages} {032110}
  (\bibinfo {year} {2007})}\BibitemShut {NoStop}%
\bibitem [{\citenamefont {Peres}(1995)}]{Peres:1995}%
  \BibitemOpen
  \bibfield  {author} {\bibinfo {author} {\bibfnamefont {A.}~\bibnamefont
  {Peres}},\ }\href@noop {} {\emph {\bibinfo {title} {Quantum Theory: Concepts
  and Methods}}}\ (\bibinfo  {publisher} {Kluwer, Dordrecht},\ \bibinfo {year}
  {1995})\BibitemShut {NoStop}%
\bibitem [{\citenamefont {Itano}\ \emph {et~al.}(1990)\citenamefont {Itano},
  \citenamefont {Heinzen}, \citenamefont {Bollinger},\ and\ \citenamefont
  {Wineland}}]{Itano:1990PRA}%
  \BibitemOpen
  \bibfield  {author} {\bibinfo {author} {\bibfnamefont {W.~M.}\ \bibnamefont
  {Itano}}, \bibinfo {author} {\bibfnamefont {D.~J.}\ \bibnamefont {Heinzen}},
  \bibinfo {author} {\bibfnamefont {J.~J.}\ \bibnamefont {Bollinger}}, \ and\
  \bibinfo {author} {\bibfnamefont {D.~J.}\ \bibnamefont {Wineland}},\
  }\bibfield  {title} {\enquote {\bibinfo {title} {Quantum {Zeno} effect},}\
  }\href {\doibase 10.1103/PhysRevA.41.2295} {\bibfield  {journal} {\bibinfo
  {journal} {Phys. Rev. A}\ }\textbf {\bibinfo {volume} {41}},\ \bibinfo
  {pages} {2295--2300} (\bibinfo {year} {1990})}\BibitemShut {NoStop}%
\bibitem [{\citenamefont {Kirchmair}\ \emph {et~al.}(2009)\citenamefont
  {Kirchmair}, \citenamefont {Z\"ahringer}, \citenamefont {Gerritsma},
  \citenamefont {Kleinmann}, \citenamefont {G\"uhne}, \citenamefont {Cabello},
  \citenamefont {Blatt},\ and\ \citenamefont {Roos}}]{Kirchmair:2009NAT}%
  \BibitemOpen
  \bibfield  {author} {\bibinfo {author} {\bibfnamefont {G.}~\bibnamefont
  {Kirchmair}}, \bibinfo {author} {\bibfnamefont {F.}~\bibnamefont
  {Z\"ahringer}}, \bibinfo {author} {\bibfnamefont {R.}~\bibnamefont
  {Gerritsma}}, \bibinfo {author} {\bibfnamefont {M.}~\bibnamefont
  {Kleinmann}}, \bibinfo {author} {\bibfnamefont {O.}~\bibnamefont {G\"uhne}},
  \bibinfo {author} {\bibfnamefont {A.}~\bibnamefont {Cabello}}, \bibinfo
  {author} {\bibfnamefont {R.}~\bibnamefont {Blatt}}, \ and\ \bibinfo {author}
  {\bibfnamefont {C.~F.}\ \bibnamefont {Roos}},\ }\bibfield  {title} {\enquote
  {\bibinfo {title} {State-independent experimental test of quantum
  contextuality},}\ }\href {\doibase 10.1038/nature08172} {\bibfield  {journal}
  {\bibinfo  {journal} {Nature (London)}\ }\textbf {\bibinfo {volume} {460}},\
  \bibinfo {pages} {494--497} (\bibinfo {year} {2009})}\BibitemShut {NoStop}%
\bibitem [{\citenamefont {L\"uders}(1951)}]{Luders:1951APL}%
  \BibitemOpen
  \bibfield  {author} {\bibinfo {author} {\bibfnamefont {G.}~\bibnamefont
  {L\"uders}},\ }\bibfield  {title} {\enquote {\bibinfo {title} {{\"U}ber die
  {Z}ustands\"anderung durch den {M}e\ss{}proze\ss},}\ }\href {\doibase
  10.1002/andp.19504430510} {\bibfield  {journal} {\bibinfo  {journal} {Ann.
  Phys. (Leipzig)}\ }\textbf {\bibinfo {volume} {443}},\ \bibinfo {pages}
  {323--328} (\bibinfo {year} {1951})}\BibitemShut {NoStop}%
\bibitem [{\citenamefont {L\"uders}(2006)}]{Luders:2006AP}%
  \BibitemOpen
  \bibfield  {author} {\bibinfo {author} {\bibfnamefont {G.}~\bibnamefont
  {L\"uders}},\ }\bibfield  {title} {\enquote {\bibinfo {title} {Concerning the
  state-change due to the measurement process},}\ }\href {\doibase
  10.1002/andp.200610207} {\bibfield  {journal} {\bibinfo  {journal} {Ann. d.
  Phys.}\ }\textbf {\bibinfo {volume} {15}},\ \bibinfo {pages} {663--670}
  (\bibinfo {year} {2006})}\BibitemShut {NoStop}%
\bibitem [{\citenamefont {Ozawa}(1984)}]{Ozawa:1984JMP}%
  \BibitemOpen
  \bibfield  {author} {\bibinfo {author} {\bibfnamefont {M.}~\bibnamefont
  {Ozawa}},\ }\bibfield  {title} {\enquote {\bibinfo {title} {Quantum measuring
  processes of continuous observables},}\ }\href {\doibase 10.1063/1.526000}
  {\bibfield  {journal} {\bibinfo  {journal} {J. Math. Phys.}\ }\textbf
  {\bibinfo {volume} {25}},\ \bibinfo {pages} {79--87} (\bibinfo {year}
  {1984})}\BibitemShut {NoStop}%
\bibitem [{\citenamefont {Heinosaari}\ and\ \citenamefont
  {Ziman}(2012)}]{Heinosaari:2012}%
  \BibitemOpen
  \bibfield  {author} {\bibinfo {author} {\bibfnamefont {T.}~\bibnamefont
  {Heinosaari}}\ and\ \bibinfo {author} {\bibfnamefont {M.}~\bibnamefont
  {Ziman}},\ }\href@noop {} {\emph {\bibinfo {title} {The mathematical language
  of quantum theory: from uncertainty to entanglement}}}\ (\bibinfo
  {publisher} {Cambridge University Press},\ \bibinfo {address} {Cambridge, New
  York},\ \bibinfo {year} {2012})\BibitemShut {NoStop}%
\bibitem [{\citenamefont {Mielnik}(1969)}]{Mielnik:1969CMP}%
  \BibitemOpen
  \bibfield  {author} {\bibinfo {author} {\bibfnamefont {B.}~\bibnamefont
  {Mielnik}},\ }\bibfield  {title} {\enquote {\bibinfo {title} {Theory of
  filters},}\ }\href {\doibase 10.1007/BF01645423} {\bibfield  {journal}
  {\bibinfo  {journal} {Comm. Math. Phys.}\ }\textbf {\bibinfo {volume} {15}},\
  \bibinfo {pages} {1--46} (\bibinfo {year} {1969})}\BibitemShut {NoStop}%
\bibitem [{\citenamefont {Araki}(1980)}]{Araki:1980CMP}%
  \BibitemOpen
  \bibfield  {author} {\bibinfo {author} {\bibfnamefont {H.}~\bibnamefont
  {Araki}},\ }\bibfield  {title} {\enquote {\bibinfo {title} {On a
  characterization of the state space of quantum mechanics},}\ }\href {\doibase
  10.1007/BF01962588} {\bibfield  {journal} {\bibinfo  {journal} {Comm. Math.
  Phys.}\ }\textbf {\bibinfo {volume} {75}},\ \bibinfo {pages} {1--24}
  (\bibinfo {year} {1980})}\BibitemShut {NoStop}%
\bibitem [{\citenamefont {Niestegge}(2008)}]{Niestegge:2008FPH}%
  \BibitemOpen
  \bibfield  {author} {\bibinfo {author} {\bibfnamefont {G.}~\bibnamefont
  {Niestegge}},\ }\bibfield  {title} {\enquote {\bibinfo {title} {A
  representation of quantum measurement in order-unit spaces},}\ }\href
  {\doibase 10.1007/s10701-008-9236-y} {\bibfield  {journal} {\bibinfo
  {journal} {Found. Phys.}\ }\textbf {\bibinfo {volume} {38}},\ \bibinfo
  {pages} {783--795} (\bibinfo {year} {2008})}\BibitemShut {NoStop}%
\bibitem [{\citenamefont {Ududec}, \citenamefont {Barnum},\ and\ \citenamefont
  {Emerson}(2011)}]{Ududec:2011FPH}%
  \BibitemOpen
  \bibfield  {author} {\bibinfo {author} {\bibfnamefont {C.}~\bibnamefont
  {Ududec}}, \bibinfo {author} {\bibfnamefont {H.}~\bibnamefont {Barnum}}, \
  and\ \bibinfo {author} {\bibfnamefont {J.}~\bibnamefont {Emerson}},\
  }\bibfield  {title} {\enquote {\bibinfo {title} {Three slit experiments and
  the structure of quantum theory},}\ }\href {\doibase
  10.1007/s10701-010-9429-z} {\bibfield  {journal} {\bibinfo  {journal} {Found.
  Phys.}\ }\textbf {\bibinfo {volume} {41}},\ \bibinfo {pages} {396--405}
  (\bibinfo {year} {2011})}\BibitemShut {NoStop}%
\bibitem [{\citenamefont {Popescu}\ and\ \citenamefont
  {Rohrlich}(1994)}]{Popescu:1994FPH}%
  \BibitemOpen
  \bibfield  {author} {\bibinfo {author} {\bibfnamefont {S.}~\bibnamefont
  {Popescu}}\ and\ \bibinfo {author} {\bibfnamefont {D.}~\bibnamefont
  {Rohrlich}},\ }\bibfield  {title} {\enquote {\bibinfo {title} {Quantum
  nonlocality as an axiom},}\ }\href {\doibase 10.1007/BF02058098} {\bibfield
  {journal} {\bibinfo  {journal} {Found. Phys.}\ }\textbf {\bibinfo {volume}
  {24}},\ \bibinfo {pages} {379--385} (\bibinfo {year} {1994})}\BibitemShut
  {NoStop}%
\bibitem [{\citenamefont {Janotta}(2012)}]{Janotta:2012EPTS}%
  \BibitemOpen
  \bibfield  {author} {\bibinfo {author} {\bibfnamefont {P.}~\bibnamefont
  {Janotta}},\ }\bibfield  {title} {\enquote {\bibinfo {title} {Generalizations
  of boxworld},}\ }\href {\doibase 10.4204/EPTCS.95.13} {\bibfield  {journal}
  {\bibinfo  {journal} {EPTCS}\ }\textbf {\bibinfo {volume} {95}},\ \bibinfo
  {pages} {183--192} (\bibinfo {year} {2012})}\BibitemShut {NoStop}%
\bibitem [{\citenamefont {Paulsen}\ and\ \citenamefont
  {Tomforde}(2009)}]{Paulsen:2009IUMJ}%
  \BibitemOpen
  \bibfield  {author} {\bibinfo {author} {\bibfnamefont {V.~I.}\ \bibnamefont
  {Paulsen}}\ and\ \bibinfo {author} {\bibfnamefont {M.}~\bibnamefont
  {Tomforde}},\ }\bibfield  {title} {\enquote {\bibinfo {title} {Vector spaces
  with an order unit},}\ }\href {\doibase 10.1512/iumj.2009.58.3518} {\bibfield
   {journal} {\bibinfo  {journal} {Indiana Univ. Math. J.}\ }\textbf {\bibinfo
  {volume} {58}},\ \bibinfo {pages} {1319--1359} (\bibinfo {year}
  {2009})}\BibitemShut {NoStop}%
\bibitem [{\citenamefont {Alfsen}(1971)}]{Alfsen:1971}%
  \BibitemOpen
  \bibfield  {author} {\bibinfo {author} {\bibfnamefont {E.~M.}\ \bibnamefont
  {Alfsen}},\ }\href@noop {} {\emph {\bibinfo {title} {Compact convex sets and
  boundary integrals}}},\ \bibinfo {series} {{Ergebnisse} {der} {Mathematik}
  {und} {ihrer} {Grenz\-ge\-biete}}\ No.~\bibinfo {number} {57}\ (\bibinfo
  {publisher} {Springer},\ \bibinfo {address} {Berlin},\ \bibinfo {year}
  {1971})\BibitemShut {NoStop}%
\bibitem [{\citenamefont {Paulsen}(2002)}]{Paulsen:2002}%
  \BibitemOpen
  \bibfield  {author} {\bibinfo {author} {\bibfnamefont {V.~I.}\ \bibnamefont
  {Paulsen}},\ }\href@noop {} {\emph {\bibinfo {title} {Completely bounded maps
  and operator algebras}}},\ \bibinfo {series} {Cambridge Studies in Advanced
  Mathematics}\ No.~\bibinfo {number} {78}\ (\bibinfo  {publisher} {Cambridge
  University Press},\ \bibinfo {address} {Cambridge},\ \bibinfo {year}
  {2002})\BibitemShut {NoStop}%
\bibitem [{\citenamefont {Ludwig}(1987)}]{Ludwig:1987}%
  \BibitemOpen
  \bibfield  {author} {\bibinfo {author} {\bibfnamefont {G.}~\bibnamefont
  {Ludwig}},\ }\href@noop {} {\emph {\bibinfo {title} {An Axiomatic Basis for
  Quantum Mechanics}}},\ Vol.\ \bibinfo {volume} {I-II}\ (\bibinfo  {publisher}
  {Springer-Verlag},\ \bibinfo {address} {Berlin},\ \bibinfo {year}
  {1987})\BibitemShut {NoStop}%
\bibitem [{\citenamefont {Mittelstaedt}(1998)}]{Mittelstaedt:1998}%
  \BibitemOpen
  \bibfield  {author} {\bibinfo {author} {\bibfnamefont {P.}~\bibnamefont
  {Mittelstaedt}},\ }\href@noop {} {\emph {\bibinfo {title} {The Interpretation
  of Quantum Mechanics and the Measurement Process}}}\ (\bibinfo  {publisher}
  {Cambridge University Press},\ \bibinfo {address} {Cambridge},\ \bibinfo
  {year} {1998})\BibitemShut {NoStop}%
\bibitem [{\citenamefont {Barnum}\ and\ \citenamefont
  {Wilce}(2011)}]{Barnum:2011ENTS}%
  \BibitemOpen
  \bibfield  {author} {\bibinfo {author} {\bibfnamefont {H.}~\bibnamefont
  {Barnum}}\ and\ \bibinfo {author} {\bibfnamefont {A.}~\bibnamefont {Wilce}},\
  }\bibfield  {title} {\enquote {\bibinfo {title} {Information processing in
  convex operational theories},}\ }\href {\doibase 10.1016/j.entcs.2011.01.002}
  {\bibfield  {journal} {\bibinfo  {journal} {Electron. Notes Theor. Comput.
  Sci.}\ }\textbf {\bibinfo {volume} {270}},\ \bibinfo {pages} {3--15}
  (\bibinfo {year} {2011})}\BibitemShut {NoStop}%
\bibitem [{\citenamefont {Alfsen}\ and\ \citenamefont
  {Shultz}(2001)}]{Alfsen:2001}%
  \BibitemOpen
  \bibfield  {author} {\bibinfo {author} {\bibfnamefont {E.~M.}\ \bibnamefont
  {Alfsen}}\ and\ \bibinfo {author} {\bibfnamefont {F.~W.}\ \bibnamefont
  {Shultz}},\ }\href@noop {} {\emph {\bibinfo {title} {State spaces of operator
  algebras: basic theory, orientations and $C*$-products}}}\ (\bibinfo
  {publisher} {Birkh\"auser},\ \bibinfo {address} {Boston},\ \bibinfo {year}
  {2001})\BibitemShut {NoStop}%
\bibitem [{\citenamefont {Barrett}(2007)}]{Barrett:2007PRA}%
  \BibitemOpen
  \bibfield  {author} {\bibinfo {author} {\bibfnamefont {J.}~\bibnamefont
  {Barrett}},\ }\bibfield  {title} {\enquote {\bibinfo {title} {Information
  processing in generalized probabilistic theories},}\ }\href {\doibase
  10.1103/PhysRevA.75.032304} {\bibfield  {journal} {\bibinfo  {journal} {Phys.
  Rev. A}\ }\textbf {\bibinfo {volume} {75}},\ \bibinfo {pages} {032304}
  (\bibinfo {year} {2007})}\BibitemShut {NoStop}%
\bibitem [{\citenamefont {Paw\l{}owski}\ and\ \citenamefont
  {Winter}(2012)}]{Pawlowski:2012PRA}%
  \BibitemOpen
  \bibfield  {author} {\bibinfo {author} {\bibfnamefont {M.}~\bibnamefont
  {Paw\l{}owski}}\ and\ \bibinfo {author} {\bibfnamefont {A.}~\bibnamefont
  {Winter}},\ }\bibfield  {title} {\enquote {\bibinfo {title} {``{Hyperbits}'':
  The information quasiparticles},}\ }\href {\doibase
  10.1103/PhysRevA.85.022331} {\bibfield  {journal} {\bibinfo  {journal} {Phys.
  Rev. A}\ }\textbf {\bibinfo {volume} {85}},\ \bibinfo {pages} {022331}
  (\bibinfo {year} {2012})}\BibitemShut {NoStop}%
\bibitem [{\citenamefont {Kleinmann}\ \emph {et~al.}(2013)\citenamefont
  {Kleinmann}, \citenamefont {Osborne}, \citenamefont {Scholz},\ and\
  \citenamefont {Werner}}]{Kleinmann:2013PRL}%
  \BibitemOpen
  \bibfield  {author} {\bibinfo {author} {\bibfnamefont {M.}~\bibnamefont
  {Kleinmann}}, \bibinfo {author} {\bibfnamefont {T.~J.}\ \bibnamefont
  {Osborne}}, \bibinfo {author} {\bibfnamefont {V.~B.}\ \bibnamefont {Scholz}},
  \ and\ \bibinfo {author} {\bibfnamefont {A.~H.}\ \bibnamefont {Werner}},\
  }\bibfield  {title} {\enquote {\bibinfo {title} {Typical local measurements
  in generalized probabilistic theories: Emergence of quantum bipartite
  correlations},}\ }\href {\doibase 10.1103/PhysRevLett.110.040403} {\bibfield
  {journal} {\bibinfo  {journal} {Phys. Rev. Lett.}\ }\textbf {\bibinfo
  {volume} {110}},\ \bibinfo {pages} {040403} (\bibinfo {year}
  {2013})}\BibitemShut {NoStop}%
\bibitem [{\citenamefont {Fritz}(2010)}]{Fritz:2010JMP}%
  \BibitemOpen
  \bibfield  {author} {\bibinfo {author} {\bibfnamefont {T.}~\bibnamefont
  {Fritz}},\ }\bibfield  {title} {\enquote {\bibinfo {title} {On the existence
  of quantum representations for two dichotomic measurements},}\ }\href
  {\doibase doi:10.1063/1.3377969} {\bibfield  {journal} {\bibinfo  {journal}
  {J. Math. Phys.}\ }\textbf {\bibinfo {volume} {51}},\ \bibinfo {pages}
  {052103} (\bibinfo {year} {2010})}\BibitemShut {NoStop}%
\bibitem [{\citenamefont {Leggett}\ and\ \citenamefont
  {Garg}(1985)}]{Leggett:1985PRL}%
  \BibitemOpen
  \bibfield  {author} {\bibinfo {author} {\bibfnamefont {A.~J.}\ \bibnamefont
  {Leggett}}\ and\ \bibinfo {author} {\bibfnamefont {A.}~\bibnamefont {Garg}},\
  }\bibfield  {title} {\enquote {\bibinfo {title} {Quantum mechanics versus
  macroscopic realism: Is the flux there when nobody looks?}}\ }\href {\doibase
  10.1103/PhysRevLett.54.857} {\bibfield  {journal} {\bibinfo  {journal} {Phys.
  Rev. Lett.}\ }\textbf {\bibinfo {volume} {54}},\ \bibinfo {pages} {857--860}
  (\bibinfo {year} {1985})}\BibitemShut {NoStop}%
\bibitem [{\citenamefont {Budroni}\ \emph {et~al.}(2013)\citenamefont
  {Budroni}, \citenamefont {Moroder}, \citenamefont {Kleinmann},\ and\
  \citenamefont {G\"uhne}}]{Budroni:2013PRL}%
  \BibitemOpen
  \bibfield  {author} {\bibinfo {author} {\bibfnamefont {C.}~\bibnamefont
  {Budroni}}, \bibinfo {author} {\bibfnamefont {T.}~\bibnamefont {Moroder}},
  \bibinfo {author} {\bibfnamefont {M.}~\bibnamefont {Kleinmann}}, \ and\
  \bibinfo {author} {\bibfnamefont {O.}~\bibnamefont {G\"uhne}},\ }\bibfield
  {title} {\enquote {\bibinfo {title} {Bounding temporal quantum
  correlations},}\ }\href {\doibase 10.1103/PhysRevLett.111.020403} {\bibfield
  {journal} {\bibinfo  {journal} {Phys. Rev. Lett.}\ }\textbf {\bibinfo
  {volume} {111}},\ \bibinfo {pages} {020403} (\bibinfo {year}
  {2013})}\BibitemShut {NoStop}%
\bibitem [{\citenamefont {Budroni}\ and\ \citenamefont
  {Emary}(2014)}]{Budroni:2014PRL}%
  \BibitemOpen
  \bibfield  {author} {\bibinfo {author} {\bibfnamefont {C.}~\bibnamefont
  {Budroni}}\ and\ \bibinfo {author} {\bibfnamefont {C.}~\bibnamefont
  {Emary}},\ }\bibfield  {title} {\enquote {\bibinfo {title} {Temporal quantum
  correlations and leggett-garg inequalities in multilevel systems},}\ }\href
  {\doibase 10.1103/PhysRevLett.113.050401} {\bibfield  {journal} {\bibinfo
  {journal} {Phys. Rev. Lett.}\ }\textbf {\bibinfo {volume} {113}},\ \bibinfo
  {pages} {050401} (\bibinfo {year} {2014})}\BibitemShut {NoStop}%
\bibitem [{\citenamefont {Sorkin}(1994)}]{Sorkin:1994MPLA}%
  \BibitemOpen
  \bibfield  {author} {\bibinfo {author} {\bibfnamefont {R.~D.}\ \bibnamefont
  {Sorkin}},\ }\bibfield  {title} {\enquote {\bibinfo {title} {Quantum
  mechanics as quantum measure theory},}\ }\href {\doibase
  10.1142/S021773239400294X} {\bibfield  {journal} {\bibinfo  {journal} {Mod.
  Phys. Lett. A}\ }\textbf {\bibinfo {volume} {9}},\ \bibinfo {pages}
  {3119--3127} (\bibinfo {year} {1994})}\BibitemShut {NoStop}%
\bibitem [{\citenamefont {Sinha}\ \emph {et~al.}(2010)\citenamefont {Sinha},
  \citenamefont {Couteau}, \citenamefont {Jennewein}, \citenamefont
  {Laflamme},\ and\ \citenamefont {Weihs}}]{Sinha:2010SCI}%
  \BibitemOpen
  \bibfield  {author} {\bibinfo {author} {\bibfnamefont {U.}~\bibnamefont
  {Sinha}}, \bibinfo {author} {\bibfnamefont {C.}~\bibnamefont {Couteau}},
  \bibinfo {author} {\bibfnamefont {T.}~\bibnamefont {Jennewein}}, \bibinfo
  {author} {\bibfnamefont {R.}~\bibnamefont {Laflamme}}, \ and\ \bibinfo
  {author} {\bibfnamefont {G.}~\bibnamefont {Weihs}},\ }\bibfield  {title}
  {\enquote {\bibinfo {title} {Ruling out multi-order interference in quantum
  mechanics},}\ }\href {\doibase 10.1126/science.1190545} {\bibfield  {journal}
  {\bibinfo  {journal} {Science}\ }\textbf {\bibinfo {volume} {329}},\ \bibinfo
  {pages} {418--421} (\bibinfo {year} {2010})}\BibitemShut {NoStop}%
\bibitem [{\citenamefont {Niestegge}(2012)}]{Niestegge:2012AMP}%
  \BibitemOpen
  \bibfield  {author} {\bibinfo {author} {\bibfnamefont {G.}~\bibnamefont
  {Niestegge}},\ }\bibfield  {title} {\enquote {\bibinfo {title} {Conditional
  probability, three-slit experiments, and the {J}ordan algebra structure of
  quantum mechanics},}\ }\href {\doibase 10.1155/2012/156573} {\bibfield
  {journal} {\bibinfo  {journal} {Adv. Math. Phys.}\ }\textbf {\bibinfo
  {volume} {2012}},\ \bibinfo {pages} {156573} (\bibinfo {year}
  {2012})}\BibitemShut {NoStop}%
\bibitem [{\citenamefont {Daki\'c}, \citenamefont {Paterek},\ and\
  \citenamefont {\v{C}. Brukner}(2014)}]{Dakic:2014NJP}%
  \BibitemOpen
  \bibfield  {author} {\bibinfo {author} {\bibfnamefont {B.}~\bibnamefont
  {Daki\'c}}, \bibinfo {author} {\bibfnamefont {T.}~\bibnamefont {Paterek}}, \
  and\ \bibinfo {author} {\bibnamefont {\v{C}. Brukner}},\ }\bibfield  {title}
  {\enquote {\bibinfo {title} {Density cubes and higher-order interference
  theories},}\ }\href {\doibase 10.1088/1367-2630/16/2/023028} {\bibfield
  {journal} {\bibinfo  {journal} {New J. Phys.}\ }\textbf {\bibinfo {volume}
  {16}},\ \bibinfo {pages} {023028} (\bibinfo {year} {2014})}\BibitemShut
  {NoStop}%
\bibitem [{\citenamefont {Alfsen}\ and\ \citenamefont
  {Shultz}(1976)}]{Alfsen:1976MAMS}%
  \BibitemOpen
  \bibfield  {author} {\bibinfo {author} {\bibfnamefont {E.~M.}\ \bibnamefont
  {Alfsen}}\ and\ \bibinfo {author} {\bibfnamefont {F.~W.}\ \bibnamefont
  {Shultz}},\ }\bibfield  {title} {\enquote {\bibinfo {title} {Non-commutative
  spectral theroy for affine function spaces on convex sets},}\ }\href
  {\doibase 10.1090/memo/0172} {\bibfield  {journal} {\bibinfo  {journal}
  {Memoirs AMS}\ }\textbf {\bibinfo {volume} {6}},\ \bibinfo {pages} {1--120}
  (\bibinfo {year} {1976})}\BibitemShut {NoStop}%
\bibitem [{\citenamefont {Alfsen}\ and\ \citenamefont
  {Shultz}(2003)}]{Alfsen:2003}%
  \BibitemOpen
  \bibfield  {author} {\bibinfo {author} {\bibfnamefont {E.~M.}\ \bibnamefont
  {Alfsen}}\ and\ \bibinfo {author} {\bibfnamefont {F.~W.}\ \bibnamefont
  {Shultz}},\ }\href@noop {} {\emph {\bibinfo {title} {Geometry of state spaces
  of operator algebras}}}\ (\bibinfo  {publisher} {Birkh\"auser},\ \bibinfo
  {address} {Boston},\ \bibinfo {year} {2003})\BibitemShut {NoStop}%
\end{thebibliography}%

\end{document}